\documentclass{eptcs}
\providecommand{\event}{EXPRESS 2010} % Name of the event you are submitting to
\def\titlerunning{Relating timed and register automata}
\def\authorrunning{D. Figueira, P. Hofman \& S. Lasota}

\usepackage{amsmath,amsthm,amscd,amssymb,amsfonts,latexsym,xspace}
\usepackage{tikz}
\usetikzlibrary{arrows}
\usetikzlibrary{automata,positioning}
\usetikzlibrary{backgrounds} % LTEX and plain TEX
\usepackage{enumerate}
\usepackage{xypic,xy}
\xyoption{all}
\usepackage{hyperref}

\theoremstyle{plain}
\newtheorem{theorem}{Theorem}[section]
\newtheorem{lemma}[theorem]{Lemma}
\newtheorem{proposition}[theorem]{Proposition}
\newtheorem{corollary}[theorem]{Corollary}

\theoremstyle{definition}

\newtheorem{example}[theorem]{Example}
\theoremstyle{remark}
\newtheorem{claim}{Claim}[theorem]

\begin{document}

% \pagestyle{plain}
% \pagenumbering{arabic}
% \markright{\today}

\title{Relating timed and register automata \thanks{Work supported by the Future and Emerging Technologies (FET) programme within the Seventh Framework Programme for Research of the European Commission, under the FET-Open grant
agreement FOX, number FP7-ICT-233599.}
%\thanks{
%This paper is partially supported by 
%Polish government grant no. N206 008 32/0810.}
}
\author{
Diego Figueira
\institute{INRIA, ENS Cachan, LSV\\France}
\and 
Piotr Hofman 
\institute{Institute of Informatics\\ University of Warsaw\\Poland}
\and 
S{\l}awomir Lasota
\institute{Institute of Informatics\\ University of Warsaw\\Poland}
}

\maketitle
 
\begin{abstract}
Timed automata and register automata are well-known  models of computation over timed and data words respectively. The former has \emph{clocks} that allow to test the lapse of time between two events, whilst the latter includes \emph{registers} 
that can store data values for later comparison. Although these two models behave in appearance  differently, several decision problems have the same (un)decidability and complexity results for both models. As a prominent example, emptiness is decidable for alternating automata
with one clock or register, both with non-primitive recursive complexity. This is not by chance.

This work confirms that there is indeed a tight relationship between the two models. We show that a run of a timed automaton can be simulated by a register automaton, and conversely that a run of a register automaton can be simulated by a timed automaton. Our results allow to transfer  complexity and decidability results back and forth between these two kinds of models. We justify the usefulness of these reductions by obtaining new results on register automata.
\end{abstract}

\newcommand{\Q}{\mathbb{Q}}
\newcommand{\Xx}{{\cal X}}
\newcommand{\Cc}{{\cal C}}
\newcommand{\Dd}{{\cal D}}
\newcommand{\confleq}{{\preceq}}
\newcommand{\nat}{{\mathbb N}}
\newcommand{\kol}{{\mathbb C}}
\newcommand{\porz}{\preceq}
\newcommand{\Data}{{\mathbb D}}
\newcommand{\doeq}{\preceq}
\newcommand{\doless}{\prec}
\newcommand{\odeq}{\succeq}
\newcommand{\odless}{\succ}
\newcommand{\lar}[1]{{\mathtt{lar}}(#1)}
\newcommand{\removed}[1]{}
\newcommand{\Alf}{\mathbb{A}}
\newcommand{\trans}[1]{\stackrel{#1}{\longrightarrow}}
\newcommand{\aut}{{\cal A}}
\newcommand{\baut}{{\cal B}}
\newcommand{\powset}[1]{{\cal P}(#1)}
\newcommand{\test}{\text{test}}
\newcommand{\tra}[4]{#1,#2,#3 \ \mapsto \ #4}
\newcommand{\Acc}{\text{F}}
\newcommand{\wqo}{\leq}
\newcommand{\wqoCc}{\sqsubseteq}
\newcommand{\rel}{\trans{}}
\newcommand{\razy}{\otimes}
\newcommand{\elt}{x}
\newcommand{\TS}{{\cal T}}
\newcommand{\ETS}{{\cal T'}}
\newcommand{\emptyw}{\varepsilon}
\newcommand{\eqperm}{\equiv}
\newcommand{\dataeq}{\sim}
\newcommand{\upto}[1]{#1^{\eqperm}}
\newcommand{\ps}[1]{{\cal P}(#1)}

\newcommand{\clo}{{\cal C}}
\newcommand{\reg}{{\cal R}}
\newcommand{\ttt}{\mathit{tt}}
\newcommand{\fff}{\mathit{ff}}
\newcommand{\deno}[1]{\lbrack #1 \rbrack}
\newcommand{\rall}{{\mathbb R}}
\newcommand{\rplus}{{\mathbb R}_+}
\newcommand{\pair}[2]{(#1,#2)}
\renewcommand{\d}{\delta}
\newcommand{\s}{\sigma}
\newcommand{\trule}[4]{\langle #1, #2, #3, #4 \rangle}
\newcommand{\tr}[4]{{#1, #2, #3 \ \mapsto \ #4}}
\newcommand{\trarr}[4]{#1, & #2, & #3 & \ \mapsto \ & #4}
\newcommand{\val}{{\mathbf v}}
\newcommand{\incr}[2]{{{#1} {+} {#2}}}
\newcommand{\reset}[2]{{{#1}[#2 := 0]}}
\newcommand{\set}[3]{{{#1}[#2 := #3]}}
\newcommand{\gamepar}[3]{{G^{#1}_{#2,#3}}}
\newcommand{\tgame}{\gamepar{\mathrm{time}}{\aut}{w}}
\newcommand{\dgame}{\gamepar{\mathrm{data}}{\aut}{w}}
\newcommand{\zero}{\val_0}
\newcommand{\conf}[2]{\pair{#1}{#2}}
\newcommand{\adam}{Adam\xspace}
\newcommand{\ewa}{Eve\xspace}
\newcommand{\tests}{\mathtt{Tests}(\reg)}

\newcommand{\constr}{\mathtt{Constr}(\clo)}    % {\Phi(\clo)}}
\newcommand{\pboo}[1]{{\cal B}^+(#1)}
\newcommand{\tik}{\checkmark}
\newcommand{\odb}[1]{\text{odb}(#1)}
\newcommand{\modb}[1]{\text{db}(#1)}
\newcommand{\mtb}[1]{\text{tb}(#1)}
\newcommand{\podloga}[1]{\lfloor #1 \rfloor}
\newcommand{\ulamek}[1]{\widehat{#1}}
\newcommand{\lang}[1]{{\cal L}(#1)}
\newcommand{\lmodb}{{\cal L}_\text{modb}}
\newcommand{\lmtb}{{\cal L}_\text{mtb}}
\newcommand{\mnote}[1]{\marginpar{\setlength{\baselineskip}{.7em}\begin{flushleft}{\tiny #1}\end{flushleft}}}
\newcommand{\pspace}{\textsc{PSpace}\xspace}
\newcommand{\midd}{~~\mid~~}
\newcommand{\coloneqq}{\mathrel{\mathop:}=}
\newcommand{\Coloneqq}{\mathrel{{\mathop:}{\mathop:}}=}
\newcommand{\tup}[1]{\langle #1 \rangle}
\newcommand\+[1]{\mathcal{#1}}
\newcommand\sset[1]{\{#1\}}

\section{Introduction}

Timed automata~\cite{AD94} and register automata
(known originally as \emph{finite-memory automata})~\cite{KF94} 
are two widely studied models of computation.
Both models extend finite automata with a kind of storage:
\emph{clocks} in the case of timed automata, capable of measuring the amount of time elapsed from the moment they were reset; and \emph{registers} in the case of register automata,
capable of storing a data value for future comparison.
In this paper we are interested in decidability and complexity 
of standard decision problems for both models of automata. In particular, we focus on the problems of \emph{nonemptiness} (Does an automaton $\aut$ accept some word?), \emph{universality} (Does an automaton $\aut$ accept all words?), and \emph{inclusion} (Are all words accepted by an automaton $\aut$ also accepted by an automaton $\baut$?).

The emptiness problem for nondeterministic
timed or register automata is \pspace-complete~\cite{AD94,DL09}. It becomes undecidable for \emph{alternating} automata of both kinds~\cite{LW08,OW05,DL09}, as soon as they have at least two clocks or registers~\cite{AD94,DL09}.
Even the universality problem was shown undecidable for nondeterministic 
timed and register automata, respectively, with two clocks or registers~\cite{AD94,NSV04,DL09}. 
A break-through result of~\cite{OW04} showed that universality becomes decidable
for one clock timed automata.
Later, the emptiness problem for one clock alternating
timed automata was shown decidable. However, the computational complexity of this problem has been 
found to be non-primitive recursive~\cite{LW08,OW05}.
Analogous (independent) results appeared for the other model: emptiness is decidable and non-primitive recursive for one register alternating automata~\cite{DL09}.
For infinite words, both one clock and one register alternating automata are undecidable, as well as the universality problem of nondeterministic one clock/register automata~\cite{LW08,ADOW05,DL09}.
The analogies between the two models appear to some extent 
also at the level of proof methods. The decidability proofs for one clock/register alternating automata are based on similar well-structured transition systems; and both non-primitive recursive lower bounds are obtained by simulation of a kind of lossy model of computation.
All these  analogies between the two models rise a natural question about the relationship between them. This paper is an attempt to answer this question.

Register automata were traditionally investigated over an unordered data domain. However, our model works on a data domain equipped with a
% dense -- removed
total order. This is a necessary extension, that allows to simulate runs of timed automata, and to have a tight equivalence between the timed and the register models. Roughly speaking, the main contribution of this paper is to show that 
timed automata and register automata over an ordered data domain are equivalent models,
as far as one concerns complexity and decidability of decision problems.

On a more technical level, we show that a run of a timed/register automaton
on a timed/data word $w$ may be simulated by a run of a register/timed automaton over a specially instrumented transformation of $w$, that we call \emph{braid}. 
The reductions we exhibit are performed in exponential time, and keep the number of clocks equal to the number of registers, and  preserve the mode of computation (alternating, nondeterministic, deterministic).
Additionally, we show that the complement of all braids is recognizable by
a nondeterministic one clock/register automaton.
These results lead straightforwardly to reductions from decision problems
for one class of automata to analogous problems for the other class, 
thus allowing us to carry over (un)decidability results and derive complexity bounds in both directions. 

As an application, our simulations allow to obtain known results on timed (or register) models as simple consequences of results on register (or timed) models. 
These include, e.g., that over finite words the emptiness problem of alternating 1 register automata is decidable~\cite{DL09}. 
In fact, our reductions yield decidability of the model extended with a 
% dense -- removed 
total order over the data domain.
As two further examples of application, we show how the following decidability results for timed 
automata can be transferred to the class of register automata:
\begin{itemize}
\item decidability of the inclusion problem between a nondeterministic (many clocks) automaton 
and an alternating one clock automaton (shown  in~\cite{LW08});
\item decidability of the emptiness problem for an alternating (many clocks) automaton over a \emph{bounded} time domain
(shown in~\cite{JORW10}).
\end{itemize}
In this paper we limit our study to \emph{finite} timed and data words, as the first step in the general program of relating the timed and data settings.

%%% Local Variables: 
%%% mode: latex
%%% TeX-master: "timedata"
%%% End: 

\section{Preliminaries}

$\rplus$ denotes the set of non-negative real numbers.
Let $\pboo{X}$ denote the set of all positive boolean formulas
over the set $X$ of propositions, i.e., the set generated by:
\begin{gather*}
  \phi \quad  \Coloneqq \quad x \midd \phi_1 \land \phi_2 \midd
\phi_1 \lor \phi_2 \qquad (x\in X) .
\end{gather*}
We fix a finite alphabet $\Alf$ for the sequel.
We recall the definitions of alternating timed and register automata~\cite{LW08,DL09}.
To avoid inessential technical complications,
we have deliberately chosen a slightly unusual definition of register automata, 
equivalent in terms of expressible power to the one defined in~\cite{DL09},
but as similar as possible to the definition of timed automata.

\subsection{Alternating timed automata}
\label{s:ata}

By a {\em timed word} over $\Alf$ we mean a finite sequence
\begin{equation}
\label{e:tw}
w ~=~ \pair{a_1}{t_1}\, \pair{a_2}{t_2}\, \ldots \,\pair{a_n}{t_n}
\end{equation}
of pairs from $\Alf \times \rplus$, with $t_1 < t_2 < \ldots < t_n$.
Each \emph{time stamp} $t_i$ 
denotes the amount of time elapsed since the beginning of the word.
For simplicity, we prefer to work with strictly monotonic timed words,
although the analogous results would hold for weakly monotonic words as well.

For a given finite set $\clo$ of {\em clock variables} (or \emph{clocks} for short), consider the set $\constr$ of clock constraints
$\sigma$ defined by
\begin{gather*}
  \sigma \quad \Coloneqq\quad  c < k \midd c \leq k \midd \sigma_1  \land \sigma_2 \midd \neg \sigma,
\end{gather*}
where
% $\rela \in \{ <, \leq \}$,
$k$ stands for an arbitrary nonnegative integer constant,
and $c \in \clo$.
For instance, note that $\ttt$ (standing for \emph{always true}), or $c = k$,
can be defined as abbreviations.
Recall also that the difference constraints $c_1 - c_2 \leq k$, typically allowed in timed
automata, may be easily eliminated
(however, the size of automaton may increase exponentially). 
A \emph{valuation} of the clocks is an element $\val\in (\rplus)^{\clo}$. 
Given a constraint $\sigma$ we write  $\deno{\sigma}$ to denote the set of clock valuations  satisfying the constraint, $\deno{\sigma} \subseteq (\rplus)^{\clo}$.

%\begin{definition}  % [Alternating timed automaton~\cite{LW08}]
%\label{d_alt}
An {\em alternating timed automaton} over $\Alf$ consists of:
a finite set  of states $Q$, a distinguished initial state $q_0 \in Q$,  
a set of accepting states $F \subseteq Q$,
a finite set $\clo$ of clocks, and 
a finite partial transition function 
$$\d: Q\times \Alf \times\constr \stackrel{\cdot}{\to} \pboo{Q \times \powset{\clo}},$$ 
subject to the following additional restriction:
\begin{description}
\item[(Partition)] For every state $q$ and label $a$, $\{\deno{\sigma} :
    \d(q,a,\s) \text{ is defined}\}$ is a (finite) partition of
    $(\rplus)^\clo$.
\end{description}
%\end{definition}

The (Partition) condition does not limit the expressive power of
automata. We impose it because it permits to give a nice symmetric
semantics for the automata as explained below. We will write
$\tr{q}{a}{\sigma}{b}$ instead of $\d(q, a, \sigma) = b$.
%Contrarily to a standard definition, we only allow to reset one clock at a time.
%This does not restrict the expressive power of timed automata
%but allows us to simplify the presentation.

To define an execution of an automaton, we will need two operations
on valuations $\val \in (\rplus)^{\clo}$.
A valuation $\incr{\val}{t}$, for $t \in \rplus$, is obtained from $\val$ by increasing the value of each clock by $t$.
A valuation $\reset{\val}{X}$, for $X \subseteq \clo$, is obtained
by reseting to zero the value of all clocks from $X$.

For an alternating timed automaton $\aut$ and a timed word $w$ as in
(\ref{e:tw}), we define the \emph{acceptance game $\tgame$} between two
players \adam\ and \ewa.  Intuitively, the objective of \ewa\ is to
accept $w$, while the aim of \adam\ is the opposite.  A play starts at
the initial configuration $\conf{q_0}{\zero}$, where $\zero : \clo \to
\rplus$ is a valuation assigning $0$ to each clock variable.  It
consists of $n$ phases.  The \mbox{$(k{+}1)$-th} phase starts in
$\conf{q_k}{\val_k}$, and ends in some configuration
$\conf{q_{k{+}1}}{\val_{k{+}1}}$ proceeding as follows.  Let
$\bar{\val} \coloneqq \val_k + t_{k{+}1} - t_k$ (for $k = 0$, $t_0$ is deemed to be $0$).  
Let $\s$ be the unique constraint
such that $\bar\val$ satisfies $\s$ and $\phi=\d(q_k,a_{k+1},\s)$ is
defined.
Existence and uniqueness of such $\s$ is implied by the (Partition)
condition.
Now the outcome of the phase is determined by
% play is guided by
the formula $\phi$. There are three cases:
\begin{itemize}
\item $\phi = \phi_1 \land \phi_2$: \adam chooses one of subformulas $\phi_1$,
$\phi_2$ and the
play continues with $\phi$ replaced by the chosen subformula;
\item
$\phi = \phi_1 \lor \phi_2$: dually, \ewa chooses one of subformulas;
\item
$\phi = (q, X) \in Q \times \powset{\clo}$: the phase ends
with the result
$\conf{q_{k{+}1}}{\val_{k{+}1}} := \conf{q}{\reset{\bar{\val}}{X}}$.
A new phase is starting from this configuration if
$k{+}1 < n$.
\end{itemize}
The winner is \ewa\ if $q_n$ is accepting ($q_n \in F$), otherwise
\adam\ wins.

Formally, a play is a finite sequence of consecutive game positions  
of the form
$\tup{ k, q, \val }$ or $\tup{ k, q, \phi }$, where $k$ is
the phase number, $\phi$ a positive boolean formula, $q$ a state and $\val$ a  
valuation.
A \emph{strategy} of \ewa is a mapping which assigns to each such sequence ending in \ewa's position a next move of \ewa. 
A strategy is \emph{winning} if \ewa wins whenever she applies this strategy.

% \begin{definition}[Acceptance]
The automaton $\aut$ {\em accepts} $w$ if{f} \ewa has a winning
strategy in the game $\tgame$.
By $\lang{\aut}$ we denote the language of all timed words $w$ accepted by
$\aut$.
% \end{definition}
%
%

\subsection{Alternating register automata}
\label{s:a1ra}

Fix an infinite data domain $\Data$.
\emph{Data words} over $\Alf$ are finite sequences
\begin{equation}
\label{e:dw}
w ~=~ \pair{a_1}{d_1} \pair{a_2}{d_2} \ldots \pair{a_n}{d_n}
\end{equation}
of pairs from $\Alf \times \Data$.  
Additionally, assume a total 
% dense -- removed
order $\doeq$ over $\Data$.
The order may be chosen arbitrarily, and our results apply to all total orders.

For a given finite set $\reg$ of \emph{register names} (or \emph{registers} for short),
consider the set $\tests$ of register tests $\sigma$ defined by
\begin{gather*}
  \sigma \quad \Coloneqq \quad \doless r \midd \doeq r \midd
  \sigma_1 \land \sigma_2 \midd \neg \sigma,\qquad \text{where $r \in \reg$.}
\end{gather*}
Each test $\sigma$ refers to registers and the current data, thus $\sigma$ denotes a subset 
$\deno{\sigma}$ of $\Data^\reg \times \Data$. E.g., $\deno{\doless r}$ means that the current data value
is strictly smaller than the value stored in register $r$.
The equality `$= r$' and inequality `$\neq r$' tests may be defined as abbreviations.

%Let the set of register operations be $\oper = \{\load_r\}_{r \in \reg}$;
%the operation $\load_r$ loads the current data value into register $r$.

An \emph{alternating register automaton} over $\Alf$ consists of:
a finite set $Q$ of states with a distinguished initial state $q_0 \in Q$ and 
a set of accepting states $F \subseteq Q$,
a finite set $\reg$ of registers, and a transition function 
\[
\d : Q \times \Alf \times \tests \to \pboo{Q {\times} \powset{\reg}}
\]
subject to the following additional restriction:
\begin{description}
\item[(Partition)] For every $q$ and $a$, the set $\{\deno{\sigma} :
    \d(q,a,\sigma) \text{ is defined}\}$ gives a (finite) partition of
    $\Data^\reg\times\Data$.
\end{description}

Register automata are typically defined over unordered data domain.
For the purpose of relating the existing models, distinguish a subclass of
register automata that only use equality $=_r$ and inequality $\neq_r$ tests;
we call them \emph{order-blind automata}. Order-blind automata correspond to the model defined in~\cite{DL09}.
% The general model we call \emph{order-sensitive} automata to avoid confusion.
As usual, we will write $\tra{q}{a}{t}{\phi}$ instead of $\d(q, a, t) = \phi$.

Given a data word $w$ as in~\eqref{e:dw}, it is accepted or not by $\aut$ depending on
the winner in the \emph{acceptance game} $\dgame$, played by \ewa\ and \adam\ 
similarly as for timed automata.
We assume for convenience that as the very first step the automaton loads 
the current data into all registers (in this way we avoid undefined values in registers).
The initial configuration is thus $\conf{q_0}{\zero}$, where $\zero : \reg \to
\Data$ assigns $d_1$ to each register.
The play consists of $n$ phases.  The \mbox{$(k{+}1)$-th} phase starts in
$\conf{q_k}{\val_k}$ and proceeds as follows.  
Let $\sigma$ be the unique test
such that $\val_k$ satisfies $\sigma$ and $\phi=\d(q_k,a_{k+1},\sigma)$ is
defined (recall the (Partition) condition).
Now the outcome of the phase is determined by
the formula $\phi$. 
The logical connectives are dealt with analogously as in case of timed automata.
When the play reaches an atomic formula $\phi = (q, X) \in Q \times \powset{\reg}$,
the phase ends with the result
$\conf{q_{k{+}1}}{\val_{k{+}1}} \coloneqq \conf{q}{\set{\val_k}{X}{d_{k+1}}}$, where
$\set{\val}{X}{d}$ differs from $\val$ by setting $\val(r) = d$ for all $r \in X$.
If $k{+}1 < n$, the game continues with a new phase starting in $\conf{q_{k{+}1}}{\val_{k{+}1}}$.

The winner is \ewa if $q_n$ is accepting ($q_n \in F$), otherwise
\adam wins.
The automaton $\aut$ {\em accepts} $w$ if{f} \ewa\ has a winning
strategy in $\dgame$.
Overloading the notation, $\lang{\aut}$ denotes the language of 
all data words accepted by $\aut$.

\removed{  %%%%%%%%%%%%%%%%%%%%%%%%%
It will be however more convenient to characterize accepting runs in terms of 
a sequence of configurations, ending in an accepting one.
A \emph{configuration} of the automaton is a finite 
set $c$ of pairs $(q, d) \in Q \times \Data$, each representing
an automaton in state $q$ with register storing the data value $d$.
In particular, an initial configuration is $c_0 = \{ (q_0, d) \}$, where $q_0$ is the initial state
of $\aut$ and data value $d$ is arbitrary.
Given a configuration $c$ and a pair $(a, d) \in \Alf \times \Data$, a possible next configuration 
is any configuration $c'$ obtained as a result of the following nondeterministic process.
(For convenience, we assume that right-hand side formulae in all transitions~\eqref{e:tra} of $\aut$
are in disjunctive normal form.)

\begin{quote}
%\paragraph{Computing $c'$ from $c$.}
Start with the empty $c'$.
For every pair $(q, e) \in c$, do the following:
\begin{itemize}
\item
Let $\phi$ be the right-hand side of the unique transition~\eqref{e:tra} 
applicable for $q, a$, depending whether $d = e$ or $d \neq e$.
\item
Choose some disjunct of $\phi$, say 
$(q_1, \load) \land \ldots (q_m, \load) \land \bar{q}_1 \land \dots \land \bar{q}_l$.
\item
Add to $c'$ the set $\{ (q_1, d) \ldots (q_m, d) \} \cup \{ (\bar{q}_1, e) \ldots (\bar{q}_l, e) \}$.
\end{itemize}
\end{quote}
We write $c \trans{(a, d)} c'$ if $c'$ may be obtained from $c$ according to the above procedure.
A data word $w \razy \mu$ is accepted, $w = a_1 \ldots a_n$,
$\mu = d_1 \ldots d_n$, if there is a sequence
\[
c_0 \trans{(a_1, d_1)} c_1 \ldots \trans{(a_n, d_n)} c_n
\]
such that $c_n$ is \emph{accepting}, i.e., only accepting states of $\aut$
appear in $c_n$.
} %%%%%%%%%%%%%%%%%% \removed

\paragraph{Deterministic, nondeterministic, and alternating.}
For both timed and register automata, we distinguish 
a subclass of nondeterministic automata as those that do not
use conjunction in the image of transition function, and a subclass of deterministic
automata that do not use disjunction either. The term alternating automata refers then to the full, unrestricted class.

\subsection{Isomorphisms}

By a \emph{time isomorphism} we mean any order-preserving bijection $f$ over the 
interval $[0, 1)$ (this implies $f(0) = 0$ in particular). 
The intuition is that an isomorphism will not be applied  to
a time stamp $t$, but to its fractional part only (that we write $\ulamek{t}\,$), keeping the 
integer part $\podloga{t}$ unchanged. %Hence, $t = \podloga{t} + \ulamek{t}$.

Given a time isomorphism $f$, we apply it to a timed word $w=(a_1,t_1) \dotsb (a_n,t_n)$ as follows:
\begin{gather*}
f(w) = \pair{a_1}{\podloga{t_1} + f(\ulamek{t_1})} \pair{a_2}{\podloga{t_2} + f(\ulamek{t_2})} \dotsb
\pair{a_n}{\podloga{t_n} + f(\ulamek{t_n})}
\end{gather*}

\begin{proposition}
\label{p:tperm}
Languages recognized by alternating timed automata are closed under time isomorphism: 
for any timed automaton $\+A$ and a time isomorphism $f$, $\+A$ accepts a timed word $w$ if{f} $\+A$ accepts $f(w)$.
\end{proposition}

%By a \emph{data isomorphism} we mean any order-preserving bijection between two finite subsets of the data domain. 
We say that two data words 
$w ~=~ \pair{a_1}{d_1} \pair{a_2}{d_2} \ldots \pair{a_n}{d_n}$ and
$v ~=~ \pair{a_1}{e_1} \pair{a_2}{e_2} \ldots \pair{a_n}{e_n}$ 
with the same string projection $a_1 a_2 \ldots a_n$
are \emph{data isomorphic} if 
for all $i, j \in \{1 \ldots n\}$, $d_i \doeq d_j$ if{f} $e_i \doeq e_j$.
%there is a data isomorphism
%that maps $d_i$ to $e_i$ for all $i = 1 \ldots n$.
% 
%The application of a data 
%isomorphism $f$  to a data word as in~\eqref{e:dw} is defined accordingly,
%\[
%f(w) = \pair{a_1}{f(d_1)}\, \pair{a_2}{f(d_2)}\, \ldots \,\pair{a_n}{f(d_n)}.
%\]
\begin{proposition}
\label{p:dperm}
Languages recognized by alternating register automata are closed under data
isomorphism: for any register automaton $\+A$ and a two data isomorphic words
$w$ and $v$,   % data isomorphism $f$, 
$\+A$ accepts $w$ if{f} $\+A$ accepts $v$.
\end{proposition}

%%% Local Variables: 
%%% mode: latex
%%% TeX-master: "timedata"
%%% End: 

\section{Braids}

An idea which is crucial to obtain reductions in both directions is an
instrumentation of timed and data words, 
to be defined in this section, that enforces a kind of `braid' structure in a word. 

\paragraph{Data braids.}

The \emph{data projection} of 
$w = \pair{a_1}{d_1} \ldots \pair{a_n}{d_n} \in (\Alf \times \Data)^*$
is $d_1  \ldots d_n \in \Data^*$.
We define the \emph{ordered partition} of a data word $w$ as a factorization 
\begin{equation}
\label{e:partition}
w_1 \cdot \ldots \cdot w_k = w
\end{equation}
into data words $w_1, \dotsc, w_k$ such that each $w_i$
is a maximal subword ordered with respect to $\doless$. 
In other words: all the data values of any $w_i$ are strictly increasing, and 
for all $i<k$, the first data value of $w_{i+1}$ is less or equal to the last one of $w_i$. 
% and if $i>1$, the last data value of $w_{i-1}$ is greater or equal to the first one of $w_i$. 
It follows that for every data word there is a unique ordered partition.
% \textcolor{red}{In my opinion second sentence is to complicated, and it only make it more difficult to understand. Piotrek}

A data word $w$ is a \emph{data braid} if{f} 
\begin{itemize}
\item The minimum data value of $w$ appears at the first position,  and
\item Its ordered partition is such that the data projection of  each factor $w_i$ is a substring of the data projection of $w_{i+1}$.  In this context, we say that $v$ is a \emph{substring} of $v'$ if{f}  $v$ is the result of removing some (possibly none) positions from  $v'$.
\item We can partition the alphabet $\Alf = \Alf_1 \cup \Alf_2$ so that a position $i$ of $w$ is labeled with a symbol of $\Alf_2$ if{f} $d_i = d_1$. 
We call a \emph{marked position} to any  $\Alf_2$-labeled position of the word. Note that the marked positions are those  starting some factor of the ordered partition of $w$.
\end{itemize}

\begin{example}
\label{e:odb}
% \textcolor{red}{We don't have ordered braid now, we have mob, so in this example should appear letters from $A_2$} 
The word $w$ below is not an ordered data braid since its ordered partition does not satisfy the substring requirement. Neither is $v$, since the minimum element does not appear at the first position.
In this example as well as in the following ones we use natural number as exemplary data value. 
\begin{align*}
w &\quad=\quad  \pair{c}{1}~~\cdot~~\pair{d}{1} \pair{a}{4} 
  \pair{b}{8}~~\cdot~~\pair{c}{1} \pair{b}{2} \pair{a}{4} \pair{a}{8}
  \pair{b}{9}~~\cdot~~\pair{c}{1},\\
v &\quad=\quad  \pair{c}{3}~~\cdot~~\pair{d}{2} \pair{a}{3} 
  \pair{b}{8}~~\cdot~~\pair{c}{2} \pair{b}{3} \pair{a}{5} \pair{a}{8}.
\end{align*}
In the case of $w$, the substring requirement is fulfilled if, e.g., the last element $\pair{c}{1}$ is removed,
or when $w$ is extended with $ \pair{b}{2} \pair{a}{4}  \pair{b}{5} \pair{a}{8} \pair{b}{9}$; in both cases
$\Alf_1 = \{a,b \}$ and $\Alf_2 = \{c, d \}$. 
\end{example}

\paragraph{Timed braids.}

Intuitively, the braid condition for timed words is analogous to that of
ordered data braids if one considers the fractional part of a time stamp $t_i$ as datum.
%Technically, the definition is different.
% We use the notation $\podloga{t}$, $t \in \rplus$, to denote the greatest integer
% smaller or equal to $t$.
A timed word 
\[
w = \pair{a_1}{t_1} \pair{a_2}{t_2} \ldots \pair{a_n}{t_n}
\]
is a \emph{timed braid} 
if the very first time stamp equals zero, $t_1 = 0$, and moreover
\begin{itemize}
\item for all $i <  n$, if $t_i < \podloga{t_n}$ then $t_i + 1$ appears
among $t_{i+1}$~\ldots~$t_n$,
\item the alphabet can be partitioned into $\Alf = \Alf_1 \cup \Alf_2$ so that the marked positions 
(i.e., those labeled by $\Alf_2$) are precisely those carrying \emph{integer} time stamp.
\end{itemize}

Braids will play a central role in the following section.
In fact both  data braids and timed braids represent essentially 
the same concept, disregarding some minor details, as illustrated next.
\begin{example}
We show  a  data braid $w$ and a `corresponding' timed braid $v$.
$\Alf_1 = \{a,b \}$ and $\Alf_2 = \{\bar{a}, \bar{b} \}$. 
\begin{align*}
w \;= \quad &  \pair{\bar{b}}{2} \pair{a}{4} & &\cdot~~~  \pair{\bar{a}}{2} \pair{b}{4}
  \pair{b}{8} &  &\cdot~~~ \pair{\bar{b}}{2} \pair{b}{3} \pair{a}{4} \pair{a}{8}
  \pair{b}{9} \\
v\;= \quad &  \pair{\bar{b}}{0.0} \pair{a}{0.5} &  &\cdot~~~ \pair{\bar{a}}{1.0} \pair{b}{1.5}
  \pair{b}{1.6} &  &\cdot~~~ \pair{\bar{b}}{2.0} \pair{b}{2.3} \pair{a}{2.5} \pair{a}{2.6}
  \pair{b}{2.9} .
\end{align*}
The particular data values and time stamps are exemplary ones.
A canonical way of obtaining a timed braid from a  data  braid
(and vice versa), to be explained below, will be ambiguous up to time (data) isomorphism.
\end{example}

\paragraph{Transformations.}

We introduce  two simple encodings:
one maps a timed word into a  data braid,
and the other maps a data word into a timed braid. 
%These simple encodings will allow us to reduce decision problems
%for timed automata to corresponding decision problems for register automata,
%and vice versa.

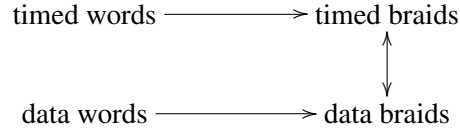
\begin{equation*}
\label{e:trans}
\xymatrix{
\text{timed words} \ar[rr] && \text{timed braids} \ar@{<->}[d] \\
\text{data words} \ar[rr] && \text{data braids} 
}
\end{equation*}

A timed word $w$ over an alphabet $\Alf$ induces 
a timed braid $\mtb{w}$ over the extended alphabet $\Alf \cup \{\tik\} \cup \bar\Alf \cup \sset{\bar\tik}$, where 
$\bar \Alf = \sset{\bar a \mid a \in \Alf}$, as follows. 
First,  if $t_1 \neq 0$, add the pair $\pair{\tik}{0}$ at the very first position.
Then add pairs $(\tik, t)$ at all time points $t$ that are missing
 according to the definition of timed braid.
 Finally change every symbol $a$ at each position carrying an integer time stamp
by its `marked' counterpart $\bar a \in \bar\Alf \cup \sset{\bar\tik}$.

A data word $w$ over $\Alf$ 
may be canonically extended to a  data braid $\modb{w}$ over the alphabet $\Alf \cup \sset{\tik} \cup \bar \Alf \cup \sset{\bar\tik}$ as follows. Consider the ordered partition $w=w_1 \cdot \ldots \cdot w_n$ and let $d_\text{min}$ be the smallest datum appearing in $w$. Firstly, for every factor $w_i$, add the pair $\pair{\tik}{d_\text{min}}$ at the very first position of $w_i$, unless $w_i$ already contains the datum $d_\text{min}$. Secondly, for each datum $d$ appearing in any $w_i$, add $\pair{\tik}{d}$ to each of the following factors $w_{i+1} \ldots w_n$ that  do not contain $d$. This insertion is done preserving the order of the factor. 
Finally, change every symbol $a$ at the first position of a factor by its `marked' counterpart $\bar a \in \bar\Alf \cup \sset{\bar\tik}$. Note that as a result we obtain a data braid. 

\begin{example}
As an illustration, consider the effect of the above transformations on an exemplary 
data word $w$ and a timed word $v$.
\begin{align*}
w &= \pair{a}{4}~~\cdot~~\pair{b}{1} \pair{a}{4} 
  \pair{b}{8}~~\cdot~~\pair{a}{1}  \pair{a}{5} \pair{a}{8} \\
\modb{w} &= \pair{\bar\tik}{1}\pair{a}{4}~~\cdot~~\pair{\bar b}{1} \pair{a}{4} 
  \pair{b}{8}~~\cdot~~\pair{\bar a}{1} \pair{\tik}{4}  \pair{a}{5}\pair{a}{8}\\
v &= \pair{a}{0.0}\pair{a}{0.7}~~\cdot~~\pair{b}{1.5} ~~\cdot~~\pair{b}{2.0}\\
\mtb{v} &= \pair{\bar a}{0.0}\pair{a}{0.7}~~\cdot~~\pair{\bar\tik}{1.0}\pair{b}{1.5}\pair{\tik}{1.7} ~~\cdot~~\pair{\bar b}{2.0}\pair{\tik}{2.5}\pair{\tik}{2.7}
\end{align*}
\end{example}

We have thus explained the horizontal arrows of the  diagram, and
now we move to the vertical ones.
Both  mappings  preserve the length of the word.

A  timed braid $\pair{a_1}{t_1} \ldots \pair{a_n}{t_n}$
gives naturally rise to a data braid by replacing
each time stamp $t_i$ by its fractional part $\ulamek{t_i}$, and then 
mapping the set $\{\ulamek{t_1}, \ldots, \ulamek{t_n}\}$ 
% interval $[0, 1)$ 
into the data domain $\Data$ through an order-preserving injection.
We only want to consider order-preserving injections, thus
this always yields a  data braid.
Note that the choice of a particular order-preserving injection is irrelevant, as 
one always obtains the same data word up to data
isomorphism (cf.~Proposition~\ref{p:dperm}).
We hope this ambiguity will not be confusing.
% in mind, denote by $\odb{w}$
% the ordered data braid obtained from a timed word $w$.

A data braid $w=(a_1,d_1) \dotsb (a_n,d_n)$ may be turned into a timed braid through any order-preserving injection 
$f : \{d_1, \ldots, d_n\} \to [0,1)$ such that $f(d_1)=0$. 
Each element $(a_i,d_i)$ is mapped into a similar element $(a_i,k+f(d_i))$, where $k$ is the number of 
factors (in the ordered partition of $w$) that end
% times that $d_1$ appears 
strictly before position $i$. Consecutive factors will get  consecutive natural numbers as the integer part of time stamps.
As before, we consider the choice of a particular injection $f$ irrelevant
(cf.~Proposition~\ref{p:tperm}).

Notice that
going from a  timed braid to a  data braid and back 
returns to the original word up to time isomorphism; likewise, combining the
transformations in the reverse order we get back to the same word,
up to data isomorphism.

Slightly overloading the notation, we write $\modb{w}$ to denote the data braid 
obtained from a \emph{timed} word $w$ by the appropriate composition of transformations just described.
Similarly, we write $\mtb{w}$ to denote the timed braid obtained from a \emph{data} word $w$.

%%% Local Variables: 
%%% mode: latex
%%% TeX-master: "timedata"
%%% End: 

\section{From timed automata to register automata}

We are going to show that, up to a suitable encoding,
languages recognized by timed automata are recognized by register automata as well. The transformation  keeps the number of registers equal to the number of
clocks, and preserve the mode of computation (nondeterministic, alternating).

\begin{theorem}
\label{t:time2reg}
Given an alternating timed automaton $\aut$ one can compute in exponential time 
an  order-blind register automaton $\baut$ such that for any timed word $w$,
$\aut$ accepts $w$ if an only if $\baut$ accepts $\modb{w}$.
The number of registers of $\baut$ equals the number of clocks of $\aut$.
Moreover, $\baut$ is deterministic (resp. nondeterministic, alternating) if $\aut$ is so.
\end{theorem}

\begin{proof}
We describe the construction of a register automaton $\baut$ that faithfully simulates  a given timed automaton $\aut$.
The idea is that the behavior of each clock can be simulated by a register. When the clock is reset on one automaton, the other loads the current data value $d$ into the register. Then, by the data braid structure, the register automaton knows exactly how many units of time have elapsed for the clock by simply counting the number of times that $d$ has appeared.

%For simplicity, in $\baut$ we will allow ourselves to use a more general simultaneous
%load operation $\load{X}$, for a nonempty subset $X \subseteq \reg$ of registers.
%This operation does not increase expressive power and may be easily eliminated.

Consider the maximum constant $k_\text{max}$ that appears in the transition rules of $\aut$. Let $Q$ and $\clo$ denote the states and clocks of $\aut$, respectively.
The states of $\baut$ will be $Q \times \{0, 1, \ldots, k_\text{max} \}^\clo$.
Intuitively, for each clock $c$ the automaton $\baut$ stores the information about 
the integer part $\podloga{c}$ of current value of $c$, up to $k_\text{max}$. In other words, $\baut$ keeps the count of how many times (up to $k_\text{max}$)  the value stored in $c$ appeared in the word since it was stored.
The initial state is $(q_0, v)$ where $v$ assigns $0$ to each $c \in \clo$.
Recall that it is assumed that as the very first step the automaton loads 
the current data into all registers.

There will be as many registers in $\baut$ as clocks in $\aut$, $\reg = \clo$, and each register
$c$ will be used to update the information about the integer part of the value of $c$.
Whenever the clock $c$ is \emph{reset} in $\aut$, the corresponding action of $\baut$ is to \emph{store} the current value 
in the register $c$ and
to change state from $(q, v)$ to $(q, v_{c})$, where $v_{c}$ 
differs from $v$ by assigning $v_{c}(c) = 0$.
The automaton $\baut$ will also be capable to detect that the integer part of $c$ increases. 
Whenever the equality test 
% \textcolor{red}{I'm not sure but i think that it should be $=c$ not $=_c$} 
$=_c$ succeeds (recall that $\baut$ is supposed to run over
a data braid) and $v(c) < k_\text{max}$, the state is changed
from $(q, v)$ to $(q, v^c)$,
where $v^c$ differs from $v$ by assigning $v^c(c) = v(c) + 1$.
On the other hand, $v$ is not changed when $v(c) = k_\text{max}$.

Now we describe the transitions of $\baut$ in more detail.
The automaton does  not distinguish marked symbols from unmarked ones.
If the current letter is $\tik$ or $\bar{\tik}$, then $\baut$ only needs to update the $v$ part of its state $(q,v)$.
Note that in this model  many registers may store the same data value, and then there will be a transition
in $\baut$ for each vector $v \in  \{0, 1, \ldots, k_\text{max} \}^\clo$ and subset $X \subseteq \clo$:
\begin{equation}
\label{e:atranmarked}
\begin{array}{c}
\tr{(q,v)}{a}{\big(\bigwedge_{c \in X} =_c \ \land \  \bigwedge_{c \notin X} \neq_c \big)}{((q, v^X), X)} \qquad\text{for $a \in \sset{\tik,\bar{\tik}}$.}
\end{array}
\end{equation}
where $v^X$ is defined by:
\[
v^X(c) = \begin{cases}
v(c) + 1 & \text{ if } c \in X \text{ and } v(c) < k_\text{max} \\
v(c) & \text{ otherwise.}
\end{cases}
\]

Otherwise, when the current letter $a$ is different from $\tik$, $\bar{\tik}$,
the automaton $\baut$ simultaneously updates $v$ similarly as above
and simulates an actual step of $\aut$.
Consider any transition 
\begin{equation}
\label{e:atran}
\tr{q}{a}{\sigma}{\phi}
\end{equation}
 of $\aut$.
We describe the corresponding transitions of $\baut$.
There are many of them,  each of them induced by 
$v$ and $X$ similarly as above. 
They are of  the following form:
\begin{equation}
\label{e:atrancorr}
\tr{(q,v)}{a}{\big(\bigwedge_{c \in X} =_c \ \land \  \bigwedge_{c \notin X} \neq_c \big)}{\phi^X_v}  \qquad\text{for $a\in\Alf \cup \bar\Alf$,}
\end{equation}
where $\phi^X_v$ is appropriately obtained from $\phi$ to ensure that
the set of clocks reset by $\aut$ is the same as the set of registers
to which $\baut$ loads the current value, and that the new
vector $v$ keeps up-to-date information about the integer parts of clock values. Let us describe  how to build $\phi^X_v$ more precisely.

Consider any fixed vector $v \in \{0, 1, \ldots, k_\text{max} \}^\clo$ and $X \subseteq \clo$. 
Being in state $(q, v)$ and reading a next letter $a$,
the automaton assumes that the clock $c$ has a value in  $(v(c), v(c)+1]$ if $v(c)<k_\text{max}$, or in $(k_\text{max},\infty)$ if $v(c)=k_\text{max}$.
Further, if a clock $c$ verifies $c \in X$, it means that the value of the corresponding register $c$ equals to the current datum. This translates in an \emph{integer} number of units of time elapsed for the clock $c$, and the exact number (up to $k_\text{max}$) is given by $v(c)+1$.
% In a similar way as  $\deno{\sigma}$ denotes a subset of $(\rplus)^\clo$,  
The pair $[v, X]$ induces a subset of $(\rplus)^\clo$  (keep in mind that the test in~\eqref{e:atrancorr} holds)
% The subset $[v, X] \subseteq  (\rplus)^\clo$ 
containing all vectors $z$ such that for each clock $c$,
\begin{align*}
  z(c) = v(c) + 1 & \text{ ~~if{f}~~ } c \in X \text{ and } v(c) < k_\text{max}, \\
  v(c) < z(c) < v(c)+1 & \text{ ~~if{f}~~ } c \notin X \text{ and  } v(c) < k_\text{max},  \\
  z(c) > k_\text{max} & \text{ ~~if{f}~~ } v(c) = k_\text{max} .
\end{align*}
Recall that $\deno{\sigma}$ is also a subset of  $(\rplus)^\clo$.
If $[v, X] \cap\deno{\sigma} \not= \emptyset$,
the transition~\eqref{e:atrancorr} is added to $\baut$.
The action $\phi^X_v$ of $\baut$ is derived from $\phi$ by replacing each pair
$(p, Y) \in Q \times \powset{\clo}$ appearing in $\phi$ with $((p, v_Y), Y)$,
where $v_Y$ is obtained from $v^X$ by setting $v_Y(c) = 0$ for
all $c \in Y$.
Note that the structure of logical connectives in $\phi^X_v$ is the same as in $\phi$. 
%Appendix~\ref{app:time2reg} depicts this translation with an example.

A careful examination of the above construction reveals that
the initial configuration of the automaton $\baut$ should be treated differently,
as no modification of $v$ should be done in this case.
We omit the details.

The automaton $\baut$ is order-blind as required.
It is deterministic (resp. nondeterministic, alternating) whenever the automaton $\aut$ is so.
The size of $\baut$ may be exponential with respect to the size of $\aut$,
as the number of different sets $X$ considered in~\eqref{e:atranmarked} and~\eqref{e:atrancorr} is exponential.
\end{proof}

\medskip

\begin{example}
As an illustration, 
consider the nondeterministic one clock timed automaton 
that checks
that there are two time stamps whose difference is $1$ depicted in Figure~\ref{fig:time-aut-checks-timestamps-diff-1}.
Nondeterminism is represented by separate arrows in the automaton
instead of disjunctive formulae. 
For the sake of clarity, we omit some (non-accepting) transitions that would
have to be added in order to fulfill the (Partition) condition.
\begin{figure}[h!]
  \centering
  \begin{center}
    \begin{tikzpicture}[>=latex, shorten >=1pt, node distance=1in, on
      grid, auto]

      \node [state, initial] (q) at (0,0) {$q$}; \node [state,
      accepting] (q1) at ( 4,0) {$p$};

      \path [->] (q) edge [loop above, %red
      ] node {
        \begin{minipage}{1cm}%
          \scriptsize{$a\\\ttt $} \\\scriptsize{$reset \:\{c\} $}
        \end{minipage}
      } (q)

      (q) edge [loop below] node {\begin{minipage}{1cm}%
          \scriptsize{$a\\\ttt $} \newline \scriptsize{$reset\:
            \emptyset$}
        \end{minipage}
      } (q)

      (q) edge %[green]
      node {
        \begin{minipage}{1cm}%
          \scriptsize{$a\\ c=1$} \\\scriptsize{$reset\: \emptyset$}
        \end{minipage}
      } (q1)

      (q1) edge [loop above] node {
        \begin{minipage}{1cm}%
          \scriptsize{$a\\ \ttt$} \\\scriptsize{$reset\: \emptyset$}
        \end{minipage}
      } (q1) ;

    \end{tikzpicture}
  \end{center}
  
  \caption{An automaton checking that there are two timestamp whose difference is 1.}
\label{fig:time-aut-checks-timestamps-diff-1}
\end{figure}
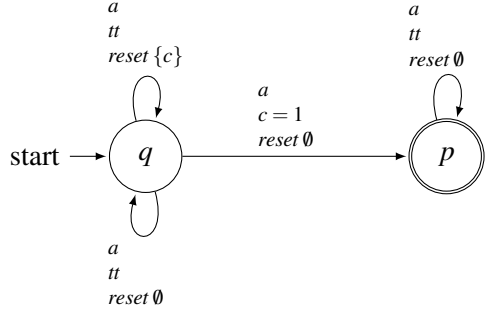
The construction described in the proof of Theorem~\ref{t:time2reg} yields the 
order-blind register automaton of Figure~\ref{fig:time2reg-automaton}.\\
\begin{figure}[h!]
  \centering
  \begin{center}
    \begin{tikzpicture}[>=latex, shorten >=1pt, node distance=1in, on
      grid, auto]
	\node [state, initial] (p0) at ( 0,0) {$q,0$};
	\node [state] (p1) at ( 4,0) {$q,1$};
	%\node [state] (p2) at ( 8-2,-6) {$q,2$};

	\node [state, accepting] (r0) at ( 8,0) {$p,0$}; 
	\node [state, accepting] (r1) at ( 12,0) {$p,1$};
 	%\node [state,accepting] (r2) at ( 8-2,-10) {$p,2$};

      \path [->] (p0) edge [loop above] node {
        \begin{minipage}{1 cm}%
          \scriptsize{$\tik , \bar{\tik} , a , \bar{a} \\ \neq_c $}
          \\\scriptsize{$load \:\emptyset $}
        \end{minipage}
      } (p0)

      (p0) edge node [pos=0.75]{
        \begin{minipage}{1 cm}%
          \scriptsize{$\tik , \bar{\tik}, a , \bar{a}\\=_c $}
          \\\scriptsize{$load \:\emptyset $}
        \end{minipage}
      } (p1)

      (p1) edge [loop below] node {
        \begin{minipage}{1.5 cm}%
          \scriptsize{$\tik , \bar{\tik}, a , \bar{a}\\ \neq_c $}
          \\\scriptsize{$load \:\emptyset $}
        \end{minipage}
      } (p1)

      (p1) edge [loop right] node [anchor=north west] {
        \begin{minipage}{1.5 cm}%
          \scriptsize{$\tik , \bar{\tik}, a , \bar{a}\\ =_c $}
          \\\scriptsize{$load \:\emptyset $}
        \end{minipage}
      } (p1)

      %(p2) edge [loop above] node {
      %  \begin{minipage}{1 cm}%
      %    \scriptsize{$\tik , \bar{\tik}, a , \bar{a}\\ \neq $}
      %    \\\scriptsize{$load \:\emptyset $}
      %  \end{minipage}
      %} (p2)

      %(p2) edge [loop right] node {
      %  \begin{minipage}{1 cm}%
      %    \scriptsize{$\tik , \bar{\tik}, a , \bar{a}\\ = $}
      %    \\\scriptsize{$load \:\emptyset $}
      %  \end{minipage}
      %} (p2)

      [->] (r0) edge [loop below] node {
        \begin{minipage}{1.5 cm}%
          \scriptsize{$\tik , \bar{\tik}, a , \bar{a}\\ \neq_c $}
          \\\scriptsize{$load \:\emptyset $}
        \end{minipage}
      } (r0)

      (r0) edge [swap] node {
        \begin{minipage}{1.5 cm}%
          \scriptsize{$\tik , \bar{\tik}, a , \bar{a}\\ =_c $}
          \\\scriptsize{$load \:\emptyset $}
        \end{minipage}
      } (r1)

      (r1) edge [loop below] node {
        \begin{minipage}{1 cm}%
          \scriptsize{$\tik , \bar{\tik}, a , \bar{a}\\ \neq_c $}
          \\\scriptsize{$load \:\emptyset $}
        \end{minipage}
      } (r1)

      (r1) edge [loop above] node {
        \begin{minipage}{1 cm}%
          \scriptsize{$\tik , \bar{\tik}, a , \bar{a}\\=_c$}
          \\\scriptsize{$load \:\emptyset $}
        \end{minipage}
      } (r1)

      %(r2) edge [loop below] node {
      %  \begin{minipage}{1 cm}%
      %    \scriptsize{$\tik , \bar{\tik}, a , \bar{a}\\\neq $}
      %    \\\scriptsize{$load \:\emptyset $}
      %  \end{minipage}
      %} (r2)

      %(r2) edge [loop right] node {
      %  \begin{minipage}{1 cm}%
      %    \scriptsize{$\tik , \bar{\tik}, a , \bar{a}\\= $}
      %    \\\scriptsize{$load \:\emptyset $}
      %  \end{minipage}
      %} (r2)

      (p0) edge [loop below
      ] node [near end]{
        \begin{minipage}{0.7 cm}%
          \scriptsize{$ a , \bar{a}\\\ttt $} \\\scriptsize{$load
            \:\{c\} $}
        \end{minipage}
      } (p0)

      (p1) edge [bend left, %red
      ] node {
        \begin{minipage}{1 cm}%
          \scriptsize{$ a , \bar{a}\\\ttt $} \\\scriptsize{$load
            \:\{c\} $}
        \end{minipage}
      } (p0)

      %(p2) edge [bend left=45, %red
      %] node [near start] {
      %  \begin{minipage}{1 cm}%
      %    \scriptsize{$ a , \bar{a}\\\ttt $} \\\scriptsize{$load
      %      \:\{c\} $}
      %  \end{minipage}
      %} (p0)

      (p0) edge [bend left=30, %green
      ] node [pos=0.5,  swap]{
        \begin{minipage}{0.5 cm}%
          \scriptsize{$ a , \bar{a}\\=_c $} \\\scriptsize{$load \:\{c\}
            $}
        \end{minipage}
      } (r1) ;

    \end{tikzpicture}
  \end{center}

\caption{The automaton resulting from the construction of Theorem~\ref{t:time2reg}.}
\label{fig:time2reg-automaton}
\end{figure}
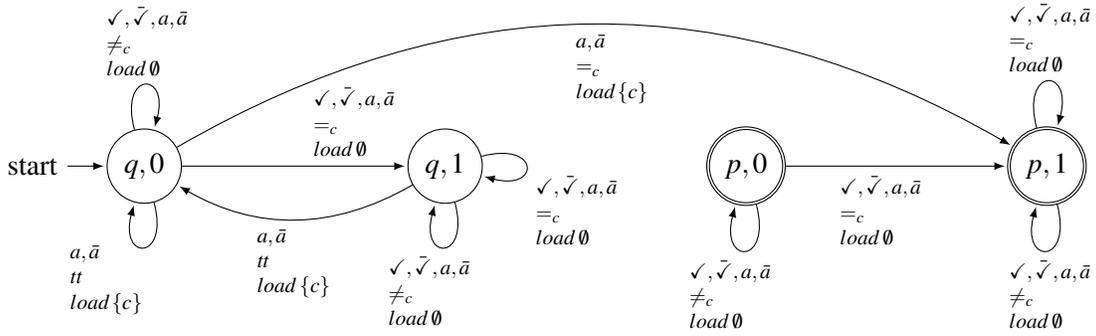
\end{example}

For the succesive results, we make use of the following lemma.

\begin{lemma}
\label{l:modb}
The complement of % $\lmodb$
the language of data braids 
is recognized by a nondeterministic one register automaton.
\end{lemma}

\begin{proof}
A data word $w=(a_1,d_1) \dotsb (a_n,d_n)$ fails to be a data braid if{f} either
\begin{itemize}
% \item \textcolor{red}{counterexample 1,5,3, 1,3,5, we should add that between to marked positions data groves.}
\item there is some marked (i.e., carrying an alphabet letter from
$\bar{\Alf} \cup \bar{\tik}$) position $i+1$ such that $d_i \doless d_{i+1}$,
\item there is some unmarked position $i+1$ such that $d_i \odeq d_{i+1}$,
\item some datum strictly smaller than $d_1$ appears in $w$, or
\item for some position $i$, there are two marked positions $j < k$, both greater than $i$,
such that $d_i$ does not appear among $\{d_j \ldots d_k\}$;
or if $d_i$ does not appear after the last marked position.
\end{itemize}
A nondeterministic automaton can easily guess which of these conditions fails and verify it using one register.
\end{proof}

As a consequence of Lemma~\ref{l:modb},
the language  of data braids is recognized by an alternating one register automaton. This is due to the fact that this model is closed under complementation.

%
%Clearly Theorem~\ref{t:time2reg} still holds with $\odb{w}$ replaced by $\modb{w}$ --
%the automaton $\baut$ simply ignores the marking.
%
We want to use Theorem~\ref{t:time2reg} together with Lemma~\ref{l:modb} to show Theorem~\ref{t:timered2reg}
below.
However, there is a subtle point here: by Lemma~\ref{l:modb} register automata can recognize
the complement of  data braids, while we would need register automata to recognize
the complement of the \emph{image} of $\modb{\_}$ (a different language, since $\modb{\_}$ is not surjective). Unfortunately, the model cannot recognize such a language. In the proof below we deal with this problem by observing that $\modb{\_}$ is \emph{essentially} surjective onto  data braids.
\begin{theorem}
\label{t:timered2reg}
The following decision problems for timed automata:
%
%\begin{center}
language inclusion,  language equality,  nonemptiness and universality,
%\end{center}
%
reduce to the analogous problems for register automata.
The reductions keep the number of registers equal to the number of clocks,
and preserve the mode of computation (nondeterministic, alternating)
of the input automaton.
\end{theorem}
\begin{proof}
Consider the inclusion problem only, the other reductions are obtained in the same way.
Given two timed automata $\aut$ and $\baut$, nondeterministic or alternating, 
we apply Theorem~\ref{t:time2reg} to obtain two corresponding register automata $\aut'$ and $\baut'$.
We claim that, for $\aut_{\neg\text{db}}$ given by Lemma~\ref{l:modb}, it holds:
% \vspace{-.3em}
\[\lang{\aut} \subseteq \lang{\baut} \text{~~~if and only if~~~} \lang{\aut'} \subseteq \lang{\baut'} \cup \lang{\aut_{\neg\text{db}}}.\]
%\vspace{-.3em}

\paragraph{(if)}
This implication is easy.
Assume $\lang{\aut'} \subseteq \lang{\baut'} \cup \lang{\aut_{\neg\text{db}}}$ and let 
$w \in \lang{\aut}$.
By Theorem~\ref{t:time2reg} $\modb{w} \in \lang{\aut'}$ and hence
$\modb{w} \in \lang{\baut'}$.
Again by Theorem~\ref{t:time2reg} $w \in \lang{\baut}$ as required.

\paragraph{(only if)}
Assume $\lang{\aut} \subseteq \lang{\baut}$ and let $w \in \lang{\aut'}$.
If $w$ is not a  data braid then $w \in \lang{\baut'} \cup \lang{\aut_{\neg\text{db}}}$ as required.
Otherwise, $w$ is a data braid, and we have the following:
\begin{claim}\label{claim:surjective-databraid}
There is a timed word $v$ such that the automata $\aut'$ and $\baut'$ cannot distinguish
between $w$ and $\modb{v}$, i.e., accept either both or none of them.
\end{claim}
With the Claim above, by Theorem~\ref{t:time2reg} we immediately obtain $v \in \lang{\aut} \subseteq \lang{\baut}$.
Again by Theorem~\ref{t:time2reg} we get $\modb{v} \in \lang{\baut'}$, thus
$w \in \lang{\baut'}$ as well due to the Claim.

\vspace{.5em}
\noindent
% \begin{proof}[
\emph{Proof of Claim~\ref{claim:surjective-databraid}.}  % ]
If the mapping $\modb{\_}$ was surjective onto data braids (up to isomorphism), then
$w$ would be equal to $\modb{v}$ (up to isomorphism) for some $v$.
This is however not the case!  For example, consider appending $(\tik,d)$ for a sufficiently big $d$ at the end of $\modb{v}$. We then obtain a data braid which is not equal to $\modb{v}$ for any $v$. But notice that in fact this last position is useless for $\aut'$ and $\baut'$.

A position $i$ in a  data braid $w$ is considered \emph{useless} if{f} (i) it is labeled by $\pair{\tik}{d}$, for some datum $d$, and
  all appearances of the datum $d$ before $i$ are labeled with $\tik$; or (ii) all the positions in its factor and in all following  factors are labeled exclusively with $\tik$ or $\bar{\tik}$.
Let $\widetilde{w}$ denote the result of removing all useless positions from $w$. We then have the following.
%\vspace{-.5em}
\begin{center}
$\widetilde{w}$ equals to $\modb{v}$, up to isomorphism, for some timed word $v$.
\end{center}
%\vspace{-.5em}
Consider any order preserving injection $f$ from data values appearing in $\widetilde{w}$ to $[0,1)$, and let $v$ be the result of the following steps:
(1) i.e., replace every $(a_i,d_i)$ of $\widetilde{w}$ with 
$(a_i, k + f(d_i))$, where $k$ is the number of factors in $\widetilde{w}$ that end before position $i$;
(2) remove all $\tik$/$\bar\tik$ positions; and (3) project the alphabet into $\Alf$.

By definition of $\modb{~}$, $\modb{v}$ is, up to isomorphism, equal to $\widetilde{w}$. Thus, $\aut'$ and $\baut'$ do not distinguish between $\widetilde{w}$ and $\modb{v}$.
It only remains to show that $\aut'$ and $\baut'$ do not distinguish between
$w$ and $\widetilde{w}$.
%\vspace{-.5em}
\begin{center}
Each of $\aut'$, $\baut'$  either accepts both $w$ and $\widetilde{w}$, or none of them.
\end{center}
%\vspace{-.5em}
This is true by construction, since when the input letter is $\tik$, the register automaton $\aut'$ (or $\baut'$) 
only updates the information about the integer part of those clocks $c$ for which the
equality test $=_c$ holds. When reading a useless position, if it is useless because of (i), then the equality  holds for no clock; whereas in case (ii) the update will be inessential for the acceptance.
%\end{proof}
%\vspace{-.5em}
This completes the proof of the claim as well as the proof of the theorem.
\end{proof}

%%% Local Variables: 
%%% mode: latex
%%% TeX-master: "timedata"
%%% End: 

\section{From register automata to timed automata}

In this section we complete the relation between the models of automata. We show that, up to a suitable encoding, languages of register automata may be
recognized by timed automata. Again, this transformation keeps the number of registers equal to the number of clocks, and preserves the mode of computation (nondeterministic, alternating).
Thus we obtain a tight relationship between the two classes of automata.

\begin{theorem}
\label{t:reg2time}
Given an alternating register automaton $\aut$ one can compute in exponential time
a timed automaton $\baut$ such that for any data word $w$,
$\aut$ accepts $w$ if and only if $\baut$ accepts $\mtb{w}$.
The number of clocks of $\baut$ equals the number of registers of $\aut$.
Moreover, $\baut$ is deterministic (resp. nondeterministic, alternating) if $\aut$ is so.
\end{theorem}

\begin{proof}
We describe the construction of a timed automaton $\baut$ that faithfully simulates
the behavior of a given register automaton $\aut$.
Let $\reg$ be the set of registers of $\aut$.
The number clocks in $\baut$ is the same as the number of registers in $\aut$, $\clo = \reg$.
A clock $r$ is reset whenever $\aut$ loads the current data value into register $r$.
Moreover, each clock is also reset whenever the constraint $r = 1$ is met.
Thus, when $\baut$ runs over a time braid,
no clock will ever have value greater than $1$.

The state space of $\baut$ is built on top of the states $Q$ of $\aut$.
Additionally, for each clock $r$ the automaton $\baut$ stores in its state one bit of 
information describing whether the last marked position was seen  before or after
the last reset of $r$.
This will allow $\baut$ to simulate tests comparing the current
data value with data values stored in registers.

Formally, states of $\baut$ are pairs $\pair{q}{X} \in Q \times \powset{\reg}$.
Initially the set $X$, which we call the register component of a state, is 
chosen as $X = \reg$ (according to the assumption that the register automaton loads the  first datum into all its registers as its very first action).
At each marked symbol $\bar{a}$ or $\bar{\tik}$, the automaton $\baut$ sets
$X := \emptyset$. 
Moreover, at each reset of $r$ (at marked or unmarked positions), $r$ is added to $X$.
As a consequence of this behavior, it invariantly holds: 
$r \in X$ if and only if the position of the last reset of $r$ is greater or equal to
the last marked position.
Hence, the test $\doeq r$ (current data smaller or equal to register $r$)
is satisfied at a state $\pair{q}{X}$ if and only if $r \notin X$.
The table below summarizes all the atomic tests and the corresponding constraints
on clock values and on the register component of state:

\begin{center}
\begin{tabular}{|c|l|l|}
\hline
\ test  in $\aut$ \ & \ meaning & \ constraint in $\baut$ \\
\hline
\ $\doeq r$ \ & \ current datum smaller or equal to $r$ \ & \ $r \notin X$ \\
$\doless r$ & \ current datum smaller than $r$ & \ $r \notin X$ and $r \neq 1$ \ \\
$\odeq r$ & \ current datum greater or equal to $r$ \ & \ $r \in X$ or $r = 1$ \\
$\odless r$ & \ current datum greater than $r$ \ & \ $r \in X$\\
\hline
\end{tabular}
\end{center}
We just described how the register component of a state is updated
and when the clocks are reset.
At input letter $\tik$ or $\bar{\tik}$ this is the only the automaton $\baut$
has to do.
Each transition $\tra{q}{a}{t}{b}$ of $\aut$ with $a \in \Alf$ gives raise to a number of transitions
$\tra{(q,X)}{a}{\sigma}{b'}$ and $\tra{(q,X)}{\bar{a}}{\sigma}{b'}$ of $\baut$
that additionally keep track of the change of state of $\aut$ and of
load operations performed by $\aut$, as described above.
The register test $t$ gives raise to a clock constraint $\sigma$ as described in
the table above.
The structure of logical connectives in $b'$ is the same as in $b$, hence
$\baut$ is deterministic (resp. nondeterministic, alternating) whenever $\aut$ is so. 
% Appendix~\ref{app:reg2time} illustrates the transformation with an example.
\end{proof}

\begin{example}
Consider the simple nondeterministic one register
automaton that checks if the first datum in a word is equal to the last one depicted in Figure~\ref{fig:reg-aut-checks-first-last-eq}.
Similarly as before, the (Partition) condition is not satisfied as
some (non-accepting) transitions are missing.
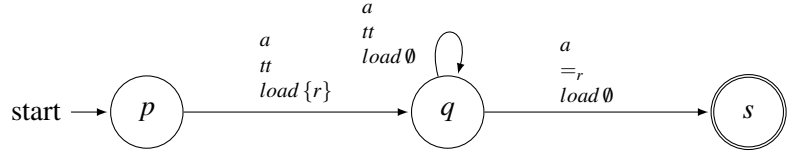
\begin{figure}[h!]
  \centering
  \begin{center}
    \begin{tikzpicture}[>=latex, shorten >=1pt, node distance=1in, on
      grid, auto]

      \node [state, initial] (q) at ( 3 * 0-2,0) {$p$}; \node [state]
      (q1) at ( 4-2,0) {$q$}; \node [state, accepting] (q2) at (
      8-2,0) {$s$};

      \path [->] (q) edge %[red]
      node {
        \begin{minipage}{1cm}%
          \scriptsize{$a\\\ttt $} \\\scriptsize{$load \: \{r\}$}
        \end{minipage}
      } (q1)

      (q1) edge [loop above] node [ anchor=east]{\begin{minipage}{1cm}%
          \scriptsize{$a\\\ttt $} \newline
          \scriptsize{$load\:\emptyset$}
        \end{minipage}
      } (q1)

      (q1) edge %[green]
      node {
        \begin{minipage}{1cm}%
          \scriptsize{$a\\ =_r$} \\\scriptsize{$load\: \emptyset$}
        \end{minipage}
      } (q2) ;

    \end{tikzpicture}
  \end{center}
  
  \caption{An automaton checking that the first datum is equal to the last one.}
\label{fig:reg-aut-checks-first-last-eq}
\end{figure}
The construction in the proof of Theorem~\ref{t:reg2time} yields the automaton of Figure~\ref{fig:reg2time-automaton}.
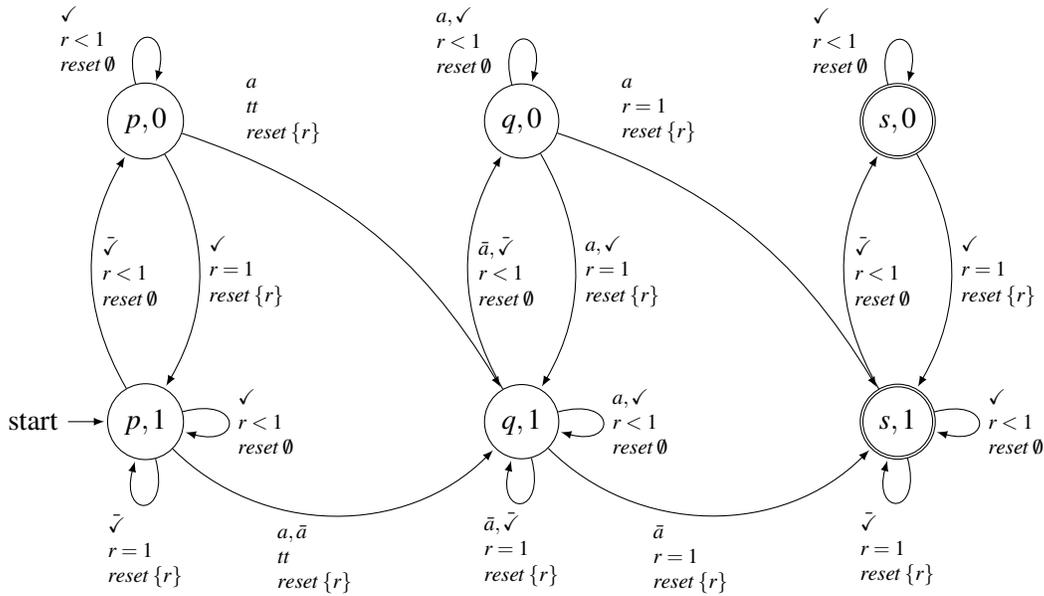
\begin{figure}[h!]
  \centering 
    \noindent
    \begin{tikzpicture}[>=latex, shorten >=1pt, node distance=1in, on
      grid, auto]

      \node [state] (q) at ( 0,0) {$p,0$}; 
      \node [state, initial] (q1) at (0, -4) {$p,1$};

      \node [state] (q0) at ( 5,0) {$q, 0$}; 
      \node [state] (q10) at (5,-4) {$q, 1$};

      \node [state, accepting ] (q0n) at ( 10,0) {$s, 0$}; 
      \node [state, accepting] (q10n) at ( 10,-4) {$s, 1$};

      \path [->] (q) edge [bend left] node [%anchor=mid east
]{\begin{minipage}{1cm}%
          \scriptsize{$\tik\\ r = 1$}\\
          \scriptsize{$reset\: \{r\}$}
        \end{minipage}
      } (q1)

      (q) edge [loop above] node [ anchor=east]{\begin{minipage}{1cm}%
          \scriptsize{$\tik\\ r < 1$} \newline \scriptsize{$reset\:
            \emptyset$}
        \end{minipage}
      } (q)
      
    (q1) edge [loop below] node {\begin{minipage}{1cm}%
          \scriptsize{$\bar{\tik}\\ r = 1$} \newline
          \scriptsize{$reset\: \{r\}$}
        \end{minipage}
      } (q1)

      (q1) edge [bend left] node [ swap%near start, anchor=center
      ]{\begin{minipage}{1.5cm}%
          \scriptsize{$\bar{\tik}\\ r < 1$}\\
          \scriptsize{$reset\: \emptyset$}
        \end{minipage}
      } (q)

      (q1) edge [loop right] node {\begin{minipage}{1cm}%
          \scriptsize{$\tik\\ r < 1$} \newline \scriptsize{$reset\:
            \emptyset$}
        \end{minipage}
      } (q1)

      (q0) edge [ bend left] node {\begin{minipage}{1cm}%
          \scriptsize{$a , \tik\\ r = 1$}\\
          \scriptsize{$reset\: \{r\}$}
        \end{minipage}
      } (q10)

      (q0) edge [loop above] node [ anchor=east]{\begin{minipage}{1cm}%
          \scriptsize{$a , \tik\\ r < 1$} \newline
          \scriptsize{$reset\: \emptyset$}
        \end{minipage}
      } (q0) 

      (q10) edge [loop below] node {\begin{minipage}{1cm}%
          \scriptsize{$\bar{a} , \bar{\tik}\\ r= 1$} \newline
          \scriptsize{$reset\: \{r\}$}
        \end{minipage}
      } (q10) 

      (q10) edge [ bend left] node [%very near end, anchor=mid east]
swap]{\begin{minipage}{1cm}%
          \scriptsize{$\bar{a} , \bar{\tik}\\ r<1$}\\
          \scriptsize{$reset\: \emptyset$}
        \end{minipage}
      } (q0)

      (q10) edge [loop right] node {\begin{minipage}{1cm}%
          \scriptsize{$a , \tik\\ r<1$} \newline \scriptsize{$reset\:
            \emptyset$}
        \end{minipage}
      } (q10)

      (q0n) edge [bend left] node {\begin{minipage}{1.5cm}%
          \scriptsize{$\tik\\ r= 1$}\\
          \scriptsize{$reset\: \{r\}$}
        \end{minipage}
      } (q10n)

      (q0n) edge [loop above] node [ anchor=east]{\begin{minipage}{1cm}%
          \scriptsize{$\tik\\ r<1$} \newline \scriptsize{$reset\:
            \emptyset$}
        \end{minipage}
      } (q0n)

      (q10n) edge [loop below] node {\begin{minipage}{1cm}%
          \scriptsize{$\bar{\tik}\\ r= 1$} \newline
          \scriptsize{$reset\: \{r\}$}
        \end{minipage}
      } (q10n)

      (q10n) edge [bend left] node [swap]{\begin{minipage}{1cm}%
          \scriptsize{$\bar{\tik}\\ r<1$}\\
          \scriptsize{$reset\: \emptyset$}
        \end{minipage}
      } (q0n)

      (q10n) edge [loop right] node {\begin{minipage}{1cm}%
          \scriptsize{$\tik\\ r<1$} \newline \scriptsize{$reset\:
            \emptyset$}
        \end{minipage}
      } (q10n)

      (q1) edge [%red
      , bend right=45, swap] node {
        \begin{minipage}{1.5cm}%
          \scriptsize{$a , \bar{a}\\ \ttt$} \\\scriptsize{$reset\:
            \{r\}$}
        \end{minipage}
      } (q10)

      (q) edge [bend left=20]
      node [very near start] {\begin{minipage}{1 cm}%
          \scriptsize{$a \\ \ttt$}\\ \scriptsize{$reset\: \{r\}$}
        \end{minipage}} (q10)

      (q0) edge [ bend left=20]
      node [very near start]{\begin{minipage}{1cm} \scriptsize{$a\\
            r=1$}\\ \scriptsize{$reset\: \{r\}$}\end{minipage} }(q10n)

      (q10) edge [%green,
      bend right=45, swap] node {\begin{minipage}{1.5 cm}
          \scriptsize{$\bar{a}\\ r=1$}\\
          \scriptsize{$reset\: \{r\}$}
        \end{minipage}}(q10n)
                 
      ;

    \end{tikzpicture}
  \caption{The result of the construction of Theorem~\ref{t:reg2time}.}
\label{fig:reg2time-automaton}
\end{figure}
\end{example}

For the next results, we make use of the following Lemma.

\begin{lemma}
\label{l:mtb}
The complement of % $\lmtb$ 
the language of  timed braids 
is recognized by a nondeterministic one clock automaton.
\end{lemma}

\begin{proof}
A timed word fails to be a timed braid if{f}: (0) the time-stamp in the first position is not equal $0$; (1) there is some marked position appearing at some non-integer position; (2) there is some unmarked position appearing at some integer position; (3) there is some missing position with time-stamp $t$ induced by the existence of an element $(a,t-1.0)$ for some $a$.

It is easy to check that (0), (1) and (2) can be verified by a deterministic one clock automaton. (3) is verified as follows. The automaton guesses a position $i$, where it resets the clock.
Then it checks that after the next marked position, the clock is continuously strictly smaller than $1$ until it finally becomes 
strictly greater than $1$, or the word ends. 
\end{proof}

As a consequence of Lemma~\ref{l:mtb},
the language  of timed braids is recognized by an alternating one clock automaton. This is due to the fact that this model is closed under complementation.

The following theorem is proved analogously to Theorem~\ref{t:timered2reg}, using Theorem~\ref{t:reg2time} and Lemma~\ref{l:mtb}:
\begin{theorem}
\label{t:regred2time}
The following decision problems for register automata:
%\begin{center}
language inclusion, language equality,  nonemptiness and universality,
%\end{center}
reduce to the analogous problems for timed automata.
The reductions keep the number of clocks equal the number of registers,
and preserve the mode of computation (nondeterministic, alternating)
of the input automata.
\end{theorem}

\begin{proof}[Sketch]
The proof is very similar to the proof of Theorem~\ref{t:timered2reg}.
We need to deal with a minor difficulty of the same kind: 
the mapping $\mtb{\_}$ is not surjective, up to isomorphism, onto  timed braids.
Thus we need the following claim, similar to one used in the proof of  
Theorem~\ref{t:timered2reg}:
\begin{claim}
For each timed braid $w$ there is a data word $v$ such that the automata 
$\aut'$ and $\baut'$ cannot distinguish
between $w$ and $\mtb{v}$.
\end{claim}
The claim is demonstrated similarly as before, by identifying \emph{useless} positions
in a  timed braid.
First, a position $i$ is \emph{useless} if{f} (1) 
it is labeled with $\pair{\tik}{t}$, for some (necessarily non-integer) $t$, and
all positions before $i$ carrying time-stamps with the same fractional part as $t$ 
are labeled with $\tik$; or  (2) position $i$ carries time stamp $t$ and all positions carrying time-stamps 
in the interval $[t, t+1)$ are labeled exclusively with $\tik$ or $\bar{\tik}$.

Similarly as before, the automaton obtained by Theorem~\ref{t:reg2time}
cannot tell the difference when all useless positions are removed from a word.
To see this, observe that in case (1) no clock will be reset as the constraint
$r = 1$ will not hold for any clock $r$.
In case (2), the whole segment of length $1$ labeled exclusively by $\tik$ and $\bar{\tik}$
may be safely skipped as it has no impact on state of the automaton.
\end{proof}

%%% Local Variables: 
%%% mode: latex
%%% TeX-master: "timedata"
%%% End: 

%\input{timeregisteraut}
\section{Applications}\label{sec:applications}

Here we provide some evidence that the tight relationship between register
automata and timed automata may be useful: we transfer a couple of results
from timed do data setting.
First, from the fact that 1-clock alternating timed automata have
decidable emptiness~\cite{LW08}, applying Theorem~\ref{t:regred2time} we obtain:
\begin{theorem}
The emptiness problem for alternating one register automata is decidable.
\end{theorem}
Note that register automata, as we define it here, work over
\emph{ordered} data domains and are capable of comparing data values w.r.t.~$\doeq$.
According to our terminology, decidability was only known for the subclass of order-blind automata~\cite{DL09}.
Interestingly, the results holds for \emph{any} total order over data.

\paragraph{A note on complexity}
Note that alternating 1 register automata over an ordered domain have the same complexity (modulo an exponential time reduction) as alternating 1 clock timed automata. However, we must remark that these automata have a much higher complexity than \emph{order blind} alternating 1 register automata, although both are beyond the primitive recursive functions. While the latter can be roughly bounded by the  Ackermann function applied to the number of states, the complexity of the former majorizes every multiply-recursive function (in particular, Ackermann's). 

More precisely, the emptiness problem for alternating timed  automata with 1 clock sits in the class $\frak F_{\omega^\omega}$ in the Fast Growing  Hierarchy \cite{fast}---an extension of the Grzegorczyk Hierarchy for non-primitive recursive functions---by a reduction to Lossy Channel Machines \cite{ADOW05}, which are known to be `complete' for this class, i.e. in $\frak F_{\omega^\omega} \setminus \frak F_{<\omega^\omega}$ \cite{CS-lics08}. 
However, the emptiness problem for \emph{order blind} alternating 1 register automata belongs to $\frak F_\omega$ in the hierarchy, by a reduction to Incrementing Counter Automata \cite{DL09}, which are complete for $\frak F_\omega$ \cite{phs-mfcs2010,FFSS10}. We then obtain the following result.
\begin{corollary}
  The emptiness problem of alternating 1 register automata over a linearly ordered domain is in $\frak F_{\omega^\omega} \setminus \frak F_{<\omega^\omega}$ in the Fast Growing Hierarchy.
\end{corollary}

\medskip

We show another example of a result that can be directly 
copied from the timed  to the data setting:
\begin{theorem}
\label{t:inclu}
Consider the following inclusion problem: Given a nondeterministic register automaton $\aut$
(with possibly many registers) and an alternating one register automaton $\baut$, 
is every word accepted by $\aut$ also accepted by $\baut$? This
problem is decidable 
and of non-primitive recursive complexity.
\end{theorem}
The result is immediate as the analogous problem is decidable and non-primitive recursive
for timed automata~\cite{LW08}.
Decidability follows from Theorem~\ref{t:regred2time}
while the lower bound follows from Theorem~\ref{t:timered2reg}.

As the last example of application of our results, we consider the emptiness problem
over a restricted class of data words. We say that a data word $\pair{a_1}{d_1} \ldots \pair{a_n}{d_n}$ is \emph{m-decreasing} 
if{f} there are at most $m-1$ positions $i$ with $d_i \odeq d_{i+1}$.
\begin{theorem}
\label{t:bound}
Let $m$ be a fixed nonnegative number.
The non-emptiness problem for alternating register automata over $m$-decreasing data words
is decidable and non-elementary.
\end{theorem}
Again, the result is an immediate consequence of the result of~\cite{JORW10}:
emptiness of alternating timed automata over $m$-bounded timed words
is decidable and non-elementary, where $m$-bounded words are those
with all time stamps smaller than $m$.
The upper bound follows by Theorem~\ref{t:reg2time} and from the fact
that a data word $w$ is $m$-decreasing if{f} $\mtb{w}$ is $m$-bounded.
Conversely, the lower bound follows from Theorem~\ref{t:time2reg} and from the 
symmetric fact: a timed word is $m$-bounded if{f} $\modb{w}$ is
$m$-decreasing.

Finally, note that the upper (decidability) bounds of Theorems~\ref{t:inclu}
and~\ref{t:bound} also apply to
the class of order-blind register automata.

%%% Local Variables: 
%%% mode: latex
%%% TeX-master: "timedata"
%%% End: 

\section{Discussion}

We have shown that timed and register automata on finite words are essentially equivalent. In order to relate these two models we introduced the notion of a `braid'-like structure, that corresponds naturally to the way a clock works when running over a timed word. As shown, most decision problems are actually equivalent for these two models. This work can be useful to derive results on one model of automata as corollaries of results on the other model, by exploiting the duality between timed and data automata shown here. This is evidenced here by showing some new results (Section~\ref{sec:applications}).
One limitation of our results is the both translations between timed and register automata
suffer of an exponential blow-up in the size of the automaton.
As a consequence, one should not expect applicability to low-complexity problems,
like non-emptiness of timed or register automata, in both cases a PSPACE-complete problem.

As a relatively straightforward further step, we plan to investigate automata over $\omega$-words; 
in particular, we hope for transferring results of~\cite{PW09,OW06} to register automata
and for comparing them e.g. with those of~\cite{L08}.
It seems however that one must restrict to dense orders on data domain for proper treatment
of infinite words.
Furthermore, we also plan to attempt a similar comparison of logical formalisms,
namely of freeze LTL and MTL or TPTL, both over finite and $\omega$-words.
This should allow to compare or transfer the complexity result for 
syntactic fragments of freeze LTL and MTL that appeared recently in 
the literature, see e.g.~\cite{PW09,OW06,L06,FS09}.

%%% Local Variables: 
%%% mode: latex
%%% TeX-master: "timedata"
%%% End: 

\vspace{-.7em}
\paragraph{Acknowledgements} The authors thank the anonymous reviewers for their valuable comments. 
\vspace{-1em}

\bibliographystyle{eptcs}
%{unsrt}%{astron}%{alpha}%apalike2}%alpha%astron}
\bibliography{timedata}

%\appendix
%\input{appendix}
\end{document}